\documentclass[11pt]{article}

\usepackage{epstopdf}

\usepackage{latexsym,amsmath,amssymb,stmaryrd}
\usepackage{graphicx,epsfig}
\usepackage[ddmmyyyy]{datetime}

\usepackage{amsthm}
\usepackage[ruled,noend,noline]{algorithm2e}
\usepackage{url}

\usepackage{multirow}
\usepackage{tabularx}
\usepackage{booktabs}

\newcolumntype{Y}{>{\raggedleft\arraybackslash}X}
\usepackage{siunitx}

\renewcommand{\epsilon}{\varepsilon}
\renewcommand{\phi}{\varphi}

\def\IN{{\mathbb N}}

\def\pref{\mathop{{\rm pref}}}
\newcommand{\floor}[1]{\lfloor #1 \rfloor}

\def\lex{{\rm lex}}
\def\CritSet{\ensuremath{\mathop{\rm CritSet}}}
\newcommand{\rank}{{\textit{rank}}}

\def\PNW{{\cal PN}}

\newcommand{\flip}{\ensuremath{\mathop{\rm flip}}}
\newcommand{\op}{\ensuremath{\mathop{\rm op}}}
\newcommand{\flipext}{\ensuremath{\mathop{\rm flipext}}}
\newcommand{\bubble}{\ensuremath{\mathop{\rm bubble}}}
\newcommand{\RightmostOne}{r}

\newtheorem{theorem}{Theorem}
\newtheorem{lemma}{Lemma}
\newtheorem{corollary}{Corollary}
\newtheorem{observation}{Observation}
\newtheorem{fact}{Fact}
\newtheorem{example}{Example}
\newtheorem{exa}[example]{Example}

\usepackage{color}

\usepackage{footmisc}
\usepackage[numbers,sort]{natbib}

\newtheorem{definition}{Definition}

\newcommand*{\email}[1]{%
	\normalsize #1 \par
}

\newcommand{\coraut}{
	\let\oldthefootnote=\thefootnote%
	\setcounter{footnote}{0}%
	\renewcommand{\thefootnote}{\fnsymbol{footnote}}%
	\footnote{Corresponding author}
	\let\thefootnote=\oldthefootnote%
}

\begin{document}

\title{Bubble-Flip---A New Generation Algorithm for Prefix Normal Words}

\author{Ferdinando Cicalese,\, Zsuzsanna Lipt{\'a}k\footnote{Corresponding author.} \,and\, Massimiliano Rossi \\ \\ \small Dipartimento di Informatica, University of Verona\\ \small Strada le Grazie, 15, 37134 Verona, Italy\\
	\email{$\{$ferdinando.cicalese,zsuzsanna.liptak,massimiliano.rossi$\_01\}$@univr.it}}

\date{\bigskip \bigskip {\em Article published in Theoretical Computer Science (2018)} \\
doi: 10.1016/j.tcs.2018.06.021}
\providecommand{\keywords}[1]{\textbf{{keywords: }} #1}
\let\oldthefootnote=\thefootnote%
\setcounter{footnote}{0}%
\renewcommand{\thefootnote}{\fnsymbol{footnote}}%
\maketitle 
\let\thefootnote=\oldthefootnote%
\begin{abstract}
We present a new recursive generation algorithm for prefix normal words. These are binary words with the property that no factor has more $1$s than the prefix of the same length. The new algorithm uses two operations on binary words, which exploit certain properties of prefix normal words in a smart way. We introduce infinite prefix normal words and show that one of the operations used by the algorithm, if applied repeatedly to extend the word, produces an ultimately periodic infinite word, which is prefix normal. Moreover, based on the original finite word, we can predict both the length and the density of an ultimate period of this infinite word\footnote{This is an extended version of our paper presented at  LATA 2018~\cite{CLR_LATA18}.}.

\medskip

\noindent \keywords{algorithms on automata and words, combinatorics on words, combinatorial generation, prefix normal words, infinite words, binary languages, combinatorial Gray code}
\end{abstract}

\section{Introduction}\label{sec:introduction}

Prefix normal words are binary words with the property that no factor has more $1$s than the prefix of the same length. For example, $11001010$ is prefix normal, but $11001101$ is not, since the factor $1101$ has too many $1$s. These words were introduced in~\cite{FL11}, originally motivated by the problem of Jumbled Pattern Matching~\cite{IJFCS12,MoosaR_JDA12,GiaGrab_IPL13,PhTRS-A14,ChanL15,AmirCLL_ICAPL14,CicaleseLWY14,AmirAHLLR16,GagieHLW15,KociumakaRR17,CunhaDGWKS17}.%

Prefix normal words have however proved to have diverse other connections~\cite{BFLRS_CPM14,BGLRS_FUN14,BFLRS_TCS17}. Among these, it has been shown that prefix normal words form a {\em bubble language}~\cite{RSW12,SW12,RSW12a}, a family of binary languages which include Lyndon words, $k$-ary Dyck words, necklaces, and other important classes of binary words. These languages have efficient generation algorithms\footnote{Here, the term {\em efficient} is used in the sense that the cost per output word should be small---in the best case, this cost is constant amortized time (CAT).}, and can be listed as (combinatorial) Gray codes, i.e.\ listings in which successive words differ by a constant number of operations. More recently, connections of the language of prefix normal words to the Binary Reflected Gray Code have been discovered~\cite{SWW17},  and prefix normal words have proved to be applicable to certain graph problems~\cite{Blondin17a}. Moreover,  three different sequences related to prefix normal words are present in the On-Line Encyclopedia of Integer Sequences (OEIS~\cite{sloane}): A194850 (the number of prefix normal words of length $n$), A238109 (a list of prefix normal words over the alphabet $\{1,2\}$), and A238110 (equivalence class sizes of words with same prefix normal form, a related concept from~\cite{BFLRS_TCS17}).

In this paper, we present a new recursive generation algorithm for prefix normal words of fixed length. In combinatorial generation, the aim is to find a way of efficiently listing (but not necessarily outputting) each one of a given class of combinatorial objects. Even though the number of these objects may be very large, typically exponential, in many situations it is necessary to be able to examine each one of them: this is when combinatorial generation algorithms are needed.
The latest volume 4A of Donald Knuth's {\em The Art of Computer Programming} devotes over 200 pages to combinatorial generation of basic combinatorial patterns, such as permutations and bitstrings~\cite{TAOCP4A}, and much more is planned on the topic~\cite{TAOCPweb}.

The previous generation algorithm for prefix normal words of length $n$
runs in amortized linear time per word~\cite{BFLRS_CPM14}, while it was conjectured there that its running time is actually amortized $O(\log n)$ per word, a conjecture which is still open. Our new algorithm recursively generates all prefix normal words from a seed word, applying two operations, which we call {\em bubble} and {\em flip}.  Its running time is $O(n)$ per word, and it allows new insights into properties of prefix normal words. It can be applied (a) to produce all prefix normal words of fixed length, or (b) to produce all prefix normal words of fixed length sharing the same {\em critical prefix}. (The critical prefix of a binary word is the first run of $1$s followed by the first run of $0$s.) This could help proving a conjecture formulated in~\cite{BFLRS_CPM14}, namely that the expected critical prefix length of an $n$-length prefix normal word is $O(\log n)$. Moreover, it could prove useful in counting prefix normal words of fixed length: it is easy to see that this number grows exponentially, however, neither a closed form nor a generating function are known~\cite{BFLRS_TCS17}. Finally, a slight change in the algorithm produces a (combinatorial) Gray code on prefix normal words of length $n$.

\medskip

While both algorithms generate prefix normal words recursively, they differ in fundamental ways. The algorithm of~\cite{BFLRS_CPM14} is an application of a general schema for generating bubble languages, using a language-specific oracle. It generates separately the sets of prefix normal words with fixed weight $d$, i.e.\ all prefix normal words of length $n$ containing $d$ $1$s. The computation tree is not binary, since each word $w$ can have up to $t$ children, where $t$ is the number of $0$s in the first run of $0$s of $w$. The algorithm uses an additional linear size data structure which it inherits from the parent node and modifies for the current node. A basic feature of the computation tree is that all words have the same fixed suffix, in other words, for the subtree rooted in the word $w = 1^s0^t\gamma$, all nodes are of the form $v\gamma$, for some $v$.

In contrast, our new algorithm generates all prefix normal words of length $n$ (except for $0^n$ and $10^{n-1}$) in one single recursive call, starting from $110^{n-2}$. The computation tree is binary, since each word can have at most two children, namely the one produced by the operation {\em bubble}, and the one by {\em flip}. Finally, for certain words $w$, the words in the subtree rooted in $w$ have the same critical prefix as $w$. This last property allows us to explore the sets of prefix normal words with fixed critical prefix.

\medskip

In the final part of the paper, we prove some surprising results about extending prefix normal words. Note that if $w$ is prefix normal, then so is $w0$, but not necessarily $w1$. We introduce infinite prefix normal words and show that repeatedly applying the flip-operation used by the new algorithm---in a modified version which {\em extends} finite words---produces, in the limit, an ultimately periodic infinite prefix normal word. Moreover, we are able to predict both the length and the density of the period, and give an upper bound on when the period will appear.

Part of the results of the present paper were presented in a preliminary form in~\cite{CLR_LATA18}.

\section{Basics}\label{sec:basics}

A (finite) binary word (or string) $w$ is a finite sequence of elements from $\{0,1\}$. We denote the $i$'th character of $w$ by $w_i$, and its length by $|w|$. Note that we index words from $1$. The empty word, denoted $\epsilon$, is the unique word with length $0$. The set of binary words of length $n$ is denoted $\{0,1\}^n$ and the set of all finite words by $\{0,1\}^*$. For two words $u,v$, we write $w=uv$ for their concatenation. For an integer $k\geq 1$ and $u\in \{0,1\}^n$, $u^k$ denotes the $k\cdot n$-length word $uuu\cdots u$ ($k$-fold concatenation of $u$). If $w=uxv$, with $u,x,v\in \{0,1\}^*$ (possibly empty), then $u$ is called a prefix, $x$ a factor (or substring), and $v$ a suffix of $w$. We denote by $w_i\cdots w_j$, for $i\leq j$, the factor of $w$ spanning the positions $i$ through $j$. For a word $u$, we write $|u|_1$ for the number of $1$s in $u$. We denote by $\leq_{\lex}$ the lexicographic order between words.

We denote by $\pref_i(w)$ the prefix of $w$ of length $i$, and by $P_w(i) = |\pref_i(w)|_1$, the number of $1$s in the prefix of length $i$. (In
the context of
succinct indexing, this function is often called $\rank_1(w,i)$.) If clear from the context, we write $P(i)$ for $P_w(i)$.

\begin{definition}[Prefix normal words, prefix normal condition]
A word $w\in \{0,1\}^*$ is called prefix normal if, for all factors $u$ of $w$, $|u|_1 \leq P_w(|u|)$. We denote the set of all finite prefix normal words by ${\cal L}$, and the set of prefix normal words of length $n$ by ${\cal L}_n$.
Given a binary word $w$, we say that a factor $u$ of $w$ {\em satisfies the prefix normal condition} if $|u|_1 \leq P_w(|u|)$.
\end{definition}

\begin{exa}
The word $1101000100110100$ is not prefix normal because the factor $1001101$ violates the prefix normal condition.
\end{exa}

It is easy to see that the number of prefix normal words grows exponentially, by noting that $1^nw$ is prefix normal for any $w$ of length $n$. In Table~\ref{tab:s2n32}, we list all prefix normal words for lengths  $n\leq 5$. Finding the number of prefix normal words of length $n$ is a challenging open problem, see~\cite{BFLRS_TCS17} for partial results. The cardinalities of ${\cal L}_n$ for $n\leq 50$ can be found in the On-Line Encyclopedia of Integer Sequences (OEIS~\cite{sloane}) as sequence A194850.

\begin{table}
	\centering
	\caption{The set ${\cal L}_n$ of prefix normal words of length $n$ for $n=1,2,3,4,5$.}
	\begin{tabularx}{\textwidth}{>{\centering}X|>{\centering}X|>{\centering}X>{\centering}X|>{\centering}X>{\centering}X>{\centering}X |>{\centering}X>{\centering}X>{\centering}X>{\centering}XX}
		\toprule
		${\cal L}_1$ & ${\cal L}_2$ & \multicolumn{2}{c|}{${\cal L}_3$} &  \multicolumn{3}{c|}{${\cal L}_4$} & \multicolumn{5}{c@{}}{${\cal L}_5$}\\
		\cmidrule(l){1-12}
		0 &  00 & 000 & 110 & 0000 & 1010 & 1110 & 00000 & 10010 & 11000 & 11011 & 11110\\
		1 &  10 & 100 & 111 & 1000 & 1100 & 1111 & 10000 & 10100 & 11001 & 11100 & 11111\\
		  &  11 & 101 &     & 1001 & 1101 &      & 10001 & 10101 & 11010 & 11101 &\\
		\bottomrule
	\end{tabularx}%
	\label{tab:s2n32}
\end{table}

Next, we give some basic facts about prefix normal words which will be needed in the following.

\begin{fact}[Basic facts about prefix normal words~\cite{BFLRS_TCS17}] \label{basicfacts}
Let $w\in\{0,1\}^n$.
\begin{enumerate}

\item[(i)] If $w\in {\cal L}$, then either $w=0^n$ or $w_1=1$.
\item[(ii)] $w \in {\cal L} $ if and only if $\pref_i(w) \in {\cal L}$ for $i = 1, \dots, n.$
\item[(iii)] If $w \in {\cal L}$ then $w0^i \in {\cal L}$ for all $i=1, 2, \dots.$
\item[(iv)] Let $w \in {\cal L}$. Then $w 1 \in {\cal L}$ if and only if for all $1\leq i< n$, we have
$P_w(i+1) > |w_{n-i+1} \cdots w_n|_1.$
\end{enumerate}
\end{fact}

We will define several operations on binary words in this paper. For an operation $\op : \{0,1\}^* \to \{0,1\}^*$, we denote by $\op^{(i)}$ the $i$'th iteration of $\op$.
We denote by $\op^*(w) = \{ \op^{(i)}(w) \mid i\geq 1\},$ the set of words obtainable from $w$ by a finite number of applications of $\op$.

Finally, we introduce the {\em critical prefix} of word. The length of the critical prefix plays an important role in the analysis of the previous generation algorithm for prefix normal words~\cite{BFLRS_CPM14}.

\begin{definition}[Critical prefix] \label{defi:criticalprefix}
Given a non-empty binary word $w$, it can be uniquely written in the form $w = 1^s0^t\gamma$, where $s,t\geq 0$, $s=0$ implies $t>0$, and $\gamma \in 1\{0,1\}^* \cup \{\epsilon\}$. We refer to $1^s0^t$ as the {\em critical prefix} of $w$.
\end{definition}

\begin{exa} For example, the critical prefix of $1100001001$ is $110000$, that of $0011101001$ is $00$, while the critical prefix of $1111000000$ is $1111000000$.
\end{exa}

In~\cite{BFLRS_CPM14}, it was conjectured that the expected length of the critical prefix of a prefix normal word of length $n$ is $O(\log n)$. This conjecture is still open. In Section~\ref{sec:critical_prefix}, we will see how to adapt our algorithm to generate all prefix normal words with critical prefix $1^s0^t$ in one run.

\medskip

To close this section, we briefly discuss {\em combinatorial Gray codes}. Recall that a {\em Gray code} is a  listing of all bitstrings (or binary words) of length $n$ such that two successive words differ by exactly one bit. In other words, a Gray code is a sequence $w^{(1)},w^{(2)},\ldots, w^{(2^n)}\in \{0,1\}^n$ such that $d_H(w^{(i)},w^{(i+1)}) = 1$ for $i=1,\ldots,2^n-1$, where $d_H(x,y) = |\{ 1\leq j \leq n : x_j \neq y_j\}|$ is the {\em Hamming distance} between two equal-length words $x$ and $y$.

This definition has been generalized in several ways, we give a definition following~\cite[ch.\ 5]{Ruskeybook}.

\begin{definition}[Combinatorial Gray Code] Given a set of combinatorial objects ${\cal S}$ and a relation $C$ on ${\cal S}$ (the closeness relation), a {\em combinatorial Gray code} for ${\cal S}$ is a listing $s_1, s_2, . . . , s_{|{\cal S}|}$ of the elements of ${\cal S}$, such that $(s_i,s_{i+1}) \in C$ for $i = 1,2,...,|{\cal S}|-1$. If we also require that $(s_{|{\cal S}|},s_1) \in C$, then the code is called {\em cyclic}.
\end{definition}

In particular, given a listing of the elements of a binary language ${\cal S}\subseteq \{0,1\}^n$, such that each two subsequent words have Hamming distance bounded by a constant, then this  listing is a combinatorial Gray code for ${\cal S}$. Note that the specifier 'combinatorial' is often dropped, so the term {\em Gray code} is frequently used in this more general sense.

\section{The Bubble-Flip algorithm}\label{sec:algorithm}

In this section we present our new generation algorithm for all prefix normal words of a given length. We show that the
words are generated in lexicographic order. We also show
how our procedure can be easily adapted to generate all prefix normal words of a given length with the same
critical prefix.

\subsection{The algorithm}

Let $w \in \{0,1\}^n.$ We let $\RightmostOne(w)$ be the largest index $r$ such that $w_r=1$,
if it exists, and $\infty$ otherwise.
We will use the following  operations on prefix normal words:

\begin{definition}[Operation \flip]\label{def:flip}
Given $w  \in \{0,1\}^n$, and $1 \leq j \leq n$, we define $\flip(w,j)$ to be the binary word obtained by changing the $j$-th character in $w$, i.e.,   $\flip(w,j) = w_1 w_2 \cdots w_{j-1} \overline{w}_j w_{j+1}\cdots w_n$, where $\overline{x}$ is $1-x$.
\end{definition}

\begin{definition}[Operation \bubble]\label{def:bubble}
 Given $w  \in \{0,1\}^n \setminus \{0^n\}$ and $r = \RightmostOne(w) < n$,  we define
 $\bubble(w) = w_1 w_2 \cdots w_{r-1}0 1 0^{n-r-1}$,
i.e., the word obtained from $w$ by shifting the rightmost $1$ one position to the right.
\end{definition}

We start by giving a simple characterization of those $\flip$-operations which preserve prefix normality.

\begin{lemma} \label{lemma:not pn}
Let $w\in {\cal L}_n$ such that $r = \RightmostOne(w)<n$  and let  $j$ be an index with
$r<j\leq n$. Then $w'= \flip(w,j)$ is not prefix normal if and only if there exists a $1 \leq k < r$ such that  $|w_{r-k+1} \cdots w_r|_1 = P_w(k)$ and $|w_{k+1} \cdots w_{k+j-r} |_1= 0$.
\end{lemma}

\begin{proof}
If there exists a $1 \leq k < r$ such that $|w_{r-k+1} \cdots w_r|_1 = P_w(k)$ and $|w_{k+1} \cdots w_{k+j-r} |_1= 0$, then for the factor $u = w'_{r-k+1} \cdots w'_j$ of $w'$, we have $|u| = k+(j-r)$ and $|u|_1 = P_{w'}(k) + 1 > P_{w'}(k+(j-r)) = P_{w'}(|u|)$, thus $w'$ is not prefix normal.

Conversely, note that $w' \in {\cal L}$ if and only if $v=\pref_j(w') \in {\cal L}$, by Fact~\ref{basicfacts} {\em (ii)} and {\em (iii)}. If $v\not\in {\cal L}$, then, by Fact~\ref{basicfacts} {\em (iv)}, there exists a suffix $u$  of $w_1\cdots w_{j-1}$ such that $|u|_1 \geq P_w(|u|+1)$
. Clearly, $u$ cannot be shorter than $j-r-1$,  since then $|u|_1 = 0 < P_w(|u|+1)$, since $w$ is prefix normal and contains at least one $1$. So $u$ spans the position $r$ of the last one of $w$. Let us write $u = u'0^{j-r-1}$, with $k := |u'|$. So we have
$P_w(k) \geq |u'|_1 = |u|_1 \geq P_w(|u|+1)$, implying $|u'|_1 = |w_{r-k+1}\cdots w_{r}|_1 = P_w(k)$ by monotonicity of $P$. Moreover, again by the monotonicity of $P$, we get $P_w(k) = P_w(|u|+1)$, which implies that the factor $w_{k+1}\cdots w_{k+j-r}$ consists of only $0$s. 
\end{proof}

\IncMargin{1em}
\begin{algorithm}
	\DontPrintSemicolon
	{Given a prefix normal word $w$, computes the leftmost index $j$, after the rightmost 1 of $w$, such that $\flip(w,j)$ is prefix normal}\;
	\LinesNumbered
	\BlankLine

\nl	$r \leftarrow \RightmostOne(w), \,  f \leftarrow 0, \, g \leftarrow 0, \,
	i \leftarrow 1, \, max \leftarrow 0$\;

\nl	\While{$i<r$}{
\nl		$f \leftarrow f + w_i, \, g \leftarrow g + w_{r-i+1}$\;
\nl		\If{$f = g$}{
\nl			$l \leftarrow 0, \, i \leftarrow i+1$\;
\nl			\While{$i < r$ and $w_i = 0$}{\label{line:found longest run of zeros}
\nl				$l \leftarrow l+1, \, i \leftarrow i+1$\;
			}
\nl			\If{$l > max$}{
\nl				$max \leftarrow  l$\;
			}
		}
\nl		\Else{
\nl				$i \leftarrow i+1$\;
		}
	}
\nl	\Return{$\min\{r+max+1, n+1\}$\;}
	\caption{{\sc Compute $\varphi$}}\label{algo:phi}
\end{algorithm}
\DecMargin{1em}

\begin{definition}[Phi]
	Let $w\in {\cal L}_n \setminus \{0^n\}$.
	 Let $r =  \RightmostOne(w)$. Define $\varphi (w)$ as the minimum $j$ such that $r < j \leq n$
	 and $\flip(w,j)$ is prefix normal, and $\varphi(w) = n+1$ if no such $j$ exists.
\end{definition}

\begin{exa}
	For the word $w=1101001001011000$, we have $\varphi(w)=16$, since the words $\flip(w,14)$ and $\flip(w,15)$ both violate the prefix normal condition, for the prefixes of length $3$ and $6$, respectively.
\end{exa}

\begin{lemma}\label{lemma:varphi is longest run of 0s}
	Let $w\in {\cal L}_n \setminus \{0^n\}$ and let $r = \RightmostOne(w)$.
	Let $m$ be the maximum length of a run of zeros following a prefix of $w_1\cdots w_{r-1}$ which has the same number of 1s 	as the suffix of $w_1\cdots w_r$ of the same length. Formally,
\[ m = \max_{1 \leq \ell < r} \{ \ell :  \text{exists } k \text{ s.t. } |w_{r-k+1} \cdots w_r|_1 = P_w(k) \text{ and }
|w_{k+1} \cdots w_{k+\ell} |_1= 0 \}, \]

\noindent where we set the maximum of the empty set to $0$.
Then, $\varphi(w) = \min(r + m + 1,n+1)$.
	\end{lemma}

\begin{proof}
	We first show that $\varphi(w) \leq r + m + 1.$ We can assume that $m < n-r,$ for otherwise the desired inequality holds by definition.
	Let $m' = m+1$. Then, there are no $j, k \in \{1, \dots, r-1\}$ such that $j-k = m', \, \, |w_1 \cdots w_k|_1 = |w_{r-k+1} \cdots w_r|_1 $ and
$ |w_{k+1} \cdots w_{j} |_1= 0.$ Thus, by Lemma \ref{lemma:not pn}, we have that $\flip(w, r+m') \in {\cal L},$ hence $\varphi(w)\leq r+m' = r+m+1.$

 \medskip
Let now $j, k$ be indices attaining the maximum in the definition of $m,$ i.e., $1<k<j<r, \, j-k = m$, $ |w_1 \cdots w_k|_1 = |w_{r-k+1} \cdots w_r|_1 $ and
$ |w_{k+1} \cdots w_{j} |_1= 0.$ Let $0 < m' \leq m$ then for $j' = k + m'$ we have
$ |w_1 \cdots w_k|_1 = |w_{r-k+1} \cdots w_r|_1 $ and
$ |w_{k+1} \cdots w_{j'} |_1= 0.$ Then, by Lemma \ref{lemma:not pn}, $\flip(w,r+m') \not \in {\cal L}.$
Hence $\varphi(w) > r+m',$ for  $m' \leq m,$ and in particular $\varphi(w) \geq r+m+1,$ which completes the proof.
 \end{proof}

Algorithm \ref{algo:phi} implements the idea of Lemma \ref{lemma:varphi is longest run of 0s} to compute $\varphi.$
For a given prefix normal word $w$, it finds the position $r$ of the rightmost $1$ in $w$.
Then, for each length $i$ such that the number of $1$s in $\pref_i(w)$ (counted by $f$) is  the same as the
number of 1s  in  $w_{r-i+1}\cdots w_r$ (counted by $g$),
the algorithm counts the number of  0s in $w$ following $\pref_i(w)$ and sets  $m$ to the maximum of the length of such runs of $0$'s. By
Lemma \ref{lemma:varphi is longest run of 0s} and the definition of $\varphi$ it follows that $\min\{r+m+1, n+1\}$ is equal to $\varphi,$ as
correctly returned by Algorithm \ref{algo:phi}. It is not hard to see that the algorithm has linear running time
since the two while-loops are only executed as long as $i < r$, and the variable $i$ increases at each iteration of either loop. Therefore, the total
number of iterations of the two loops together is upper bounded by $r \leq n.$ Thus, we have proved the following lemma:

\begin{lemma} \label{lemma:phicomputationlemma}
For $w \in {\cal L}_n \setminus \{0^n\},$ Algorithm \ref{algo:phi} computes $\varphi(w)$ in
$O(\RightmostOne(w))$, hence $O(n)$ time.
\end{lemma}

The next lemma gives the basis of our algorithm: applying either of the two operations $\flip(w,\phi(w))$ or $\bubble(w)$ to a prefix normal word $w$ results in another prefix normal word.

\begin{lemma}
	Let $w \in {\cal L}_n \setminus \{0^n\}$. Then the following holds:
	\begin{itemize}
		\item[a)] for every $\ell$, such that $\varphi(w)  \leq \ell \leq n$, $\flip(w,\ell)$ is prefix normal, and
		\item[b)] if $|w|_1 \geq 2$ then $\bubble(w)$ is prefix normal.

	\end{itemize}
\end{lemma}

\begin{proof}
	Let $r = \RightmostOne(w)$. In order to show {\em a)} we can proceed as in the proof of the upper bound in
	Lemma \ref{lemma:varphi is longest run of 0s}. Fix $\varphi(w) \leq \ell \leq n,$ and let $m' = \ell - r.$
	Then, by Lemma \ref{lemma:varphi is longest run of 0s}, there exist no $1 < j < k < r$ such the $k - j = m'$ and
	$|w_1 \cdots w_k|_1 = |w_{r-k+1} \cdots w_r|_1 $ and
$ |w_{k+1} \cdots w_{j} |_1= 0.$ This, by Lemma \ref{lemma:not pn}, implies that $\flip(w, \ell) \in {\cal L}.$

For {\em b)}, let $r' = \max\{i < r \mid w_i = 1\},$ i.e., $r'$ is the position of the penultimate $1$ of $w$. Let
$w' = w_1 \cdots w_{r'}0^{n-r'}.$ By Fact \ref{basicfacts} we have that $w' \in {\cal L}.$  Moreover,
$r \geq \varphi(w'),$ since $\flip(w', r)=w \in {\cal L}.$ Therefore, by {\em a)} we have that $\bubble(w) = \flip(w', r+1) \in {\cal L}.$
 \end{proof}

\begin{definition}[$\PNW$]\label{def:descent}
	Given $w\in {\cal L}_n \setminus \{0^n\}$ with $r = \RightmostOne(w)$, we define $\PNW(w)$ as the set of all prefix normal words $v$ of length $n$ such that $v = w_1 \cdots w_{r-1} \gamma$ for some $\gamma$ with $|\gamma|_1 > 0.$ Formally,
	 $$ \PNW(w) = \{v \in {\cal L}_n \mid  v = w_1 \cdots w_{r-1} \gamma , %
	 |\gamma|_1 >0 \}.$$
We will use the convention  that  $\PNW(\flip(w,\varphi(w)))=\emptyset$ if $\varphi(w) > n$, and $\PNW(\bubble(w)) = \emptyset$ if $\RightmostOne(w) =n$, since then $\flip(w,\varphi(w))$ resp.\ $\bubble(w)$
 are undefined.
\end{definition}

\begin{lemma}\label{lemma:recursion}\label{disjointness}
Given $w\in {\cal L}_n \setminus \{0^n,10^{n-1}\}$, we have
\[\PNW(w) =\{w\} \cup \PNW(\flip(w,\varphi(w))) \cup \PNW(\bubble(w)).\]

Moreover, these three sets are pairwise disjoint.
\end{lemma}

\begin{proof} It is easy to see that the sets $\{w\}$, $\PNW(\bubble(w))$, $\PNW(\flip(w,$ $\varphi(w)))$ are pairwise disjoint.

The inclusion $\PNW(w) \supseteq \{w\} \cup \PNW(\flip(w,\varphi(w))) \cup \PNW(\bubble(w))$ follows from the definition
of $\PNW$ (Def.~\ref{def:descent}) for each of the words $w, \flip(w,\phi(w)),$ and $\bubble(w)$.

Now let $x \in \PNW(w)\setminus\{w\}$ and $r = \RightmostOne(w).$
We argue by cases according to the character $x_r.$

{\em Case 1.} $x_r = 0.$ Then, $x = w_1\cdots w_{r-1} 0 \gamma$  for some $\gamma \in \{0,1\}^{n-r}$ such that
$|\gamma|_1 > 0.$ Since $\bubble(w) = w_1\cdots w_{r-1} 0 1 0^{n-r-1},$ it follows  that $x \in \PNW(\bubble(w)).$

{\em Case 2.} $x_r = 1.$ Then, since $x \neq w,$ we also have that $|x_{r+1}\cdots x_n|_1 > 0.$
Therefore,  $x = w_1\cdots w_{r-1} 1 \gamma$  for some $\gamma \in \{0,1\}^{n-r}$ such that
$|\gamma|_1 > 0.$

Let $r' = \min\{i > r \mid x_{r'} = 1\}$. Since $x \in {\cal L},$ we have that $\pref_r(x)0^{n-r},$ $\pref_{r'}(x)0^{n-r'} \in {\cal L}.$
Moreover, $\pref_{r'}(x)0^{n-r'} = \flip(\pref_r(x)0^{n-r}, r'),$ hence, $r' \geq \varphi(\pref_r(x)0^{n-r}) = \varphi(w).$
Therefore, $x = w_1\cdots w_r0^{\varphi(w)-r -1} \gamma$ for some $|\gamma|_1 > 1.$ This, by definition,  means that
$x \in \PNW(\flip(w, \varphi(w))).$

\end{proof}

We are now ready to give an algorithm computing all words in the set $\PNW(w)$ for a prefix normal word $w$. The pseudocode is given in Algorithm~\ref{algo:pnw_descent}.
The procedure generates recursively the set $\PNW(w)$ as the union of
$\PNW(\flip(w, \varphi(w)))$ and $\PNW(\bubble(w))$.
The call to subroutine $Visit()$ is a placeholder indicating that the algorithm has generated a new word in $\PNW(w)$, which could be printed, or examined, or processed, as required. By Lemma \ref{disjointness} we know that $Visit()$ is executed for each word in $\PNW(w)$ exactly once.

\IncMargin{1em}
\begin{algorithm}
	\DontPrintSemicolon
	{Given a prefix normal word $w$ such that $|w|_1 > 1$, generate the set $\PNW(w)$}\;
	\BlankLine

\nl	\If{$\RightmostOne(w) \neq n$}{
\nl		$w^{\prime} = \bubble(w)$\;
\nl		{\sc Generate } $\PNW (w^{\prime})$\;
	}
\nl	$Visit()$\;
\nl	$j = \varphi(w)$\;
\nl	\If{$j\leq n$}{
\nl		$w^{\prime\prime} = \flip(w,j)$\;
\nl		{\sc Generate } $\PNW(w^{\prime\prime})$\;
	}

	\caption{{\sc Generate} $\PNW(w)$}\label{algo:pnw_descent}
\end{algorithm}
\DecMargin{1em}

In order to ease the running time analysis, we next introduce a tree ${\cal T}(w)$ on $\PNW(w)$. This tree coincides with the computation tree of {\sc Generate} $\PNW(w)$, but it will be useful to argue about it independently of the algorithm.

\begin{definition}[Tree on $\PNW(w)$]
Let $w\in {\cal L}_n \setminus \{0^n,10^{n-1}\}$. Then we denote by ${\cal T}(w)$ the rooted binary tree ${\cal T}$ with $V({\cal T}) = \PNW(w)$, root $w$, and for a node $v$, (1) the left child of $v$ defined as empty if $v_n=1$ and as $\bubble(v)$ otherwise, and (2) the right child of $v$ as empty if $\phi(v) = n+1$, and as $\flip(v,\phi(v))$ otherwise.
\end{definition}

The tree $\PNW(w)$ has the following easy-to-see properties.

\begin{observation}[Properties of ${\cal T}(w)$]\label{obs:tree_properties}
There are three types of nodes: the root $w$, \bubble-nodes (left children), and \flip-nodes (right children).

\begin{enumerate}
\item The leftmost descendant of $w$ has maximal depth, namely $n-r$, where $r = \RightmostOne(w)$.
\item If a node $v$ has a right child, then it also has a left child.
\item If a node $v$ has no right child, then no descendant of $v$ has a right child. Thus in this case, the subtree rooted in $v$ is a path of length $n-r'$, consisting only of \bubble-nodes, where $r' = \RightmostOne(v)$.
\end{enumerate}
\end{observation}

The next lemma gives correctness, the generation order, and running time of algorithm {\sc Generate} $\PNW(w)$.

	\begin{lemma}\label{lemma:lex}
		For $w\in {\cal L}_n \setminus \{0^n,10^{n-1}\}$,  Algorithm~\ref{algo:pnw_descent} generates all prefix normal words in $\PNW(w)$ in
		lexicographic order in $O(n)$ time per word.
	\end{lemma}

	\begin{proof}
		Algorithm \ref{algo:pnw_descent} recursively generates first all words in
		$\PNW(\bubble(w)))$, then the word $w$, and finally the words in $\PNW(\flip(w, \varphi(w)))$. As we saw above (Lemma~\ref{lemma:recursion}), these sets form a partition of $\PNW(w)$, hence
		every word $v \in \PNW(w)$  is generated exactly once.
		Moreover, by definition of $\PNW$, for every $u \in \PNW(\bubble(w))$ it holds that  $u = w_1 \cdots w_{r-1}0\gamma$ with $|\gamma| = n-r $ and $ |\gamma|_1 >0$, thus it follows that $ u <_{\lex} w$. In addition,  for every $v \in \PNW(\flip(w,\varphi(w)))$ it holds that $v =  w_1 \cdots w_{r-1}1\beta \gamma $ where $ |\beta| = k = \varphi(w)-r-1 $, $  |\beta|_1 =0 ,$ $|\gamma| = n-r-k $ and $ |\gamma|_1 >0$, thus $ w <_{\lex} v$. Since these relations hold at every level of the recursion, it follows that the words are generated by Algorithm \ref{algo:pnw_descent} in lexicographic order.

For the running time, note that in each node $v$, the algorithm spends $O(n)$ time on the computation of $\phi(v)$ (Lemma~\ref{lemma:phicomputationlemma}), and if $v_n\neq 1$, another $O(1)$ time on computing $\bubble(v)$, and finally, if $\phi(v) \leq n$, further $O(1)$ time on computing $\flip(v,\phi(v))$. This gives a total running time of $O(n\cdot\PNW(w))$, so $O(n)$ amortized time per word. We now show that it actually runs in $O(n)$ time per word.

Notice that the algorithm performs an in-order traversal of the tree ${\cal T}(w)$. Given a node $v$, the next node visited by the algorithm is given by:
\[
next(v) = \begin{cases}
\text{leftmost descendant of right child,} & \text{if $\phi(v) \leq n$},\\
parent(v), & \hspace{-3.6cm} \text{if $\phi(v) > n$ and $v$ is a left child},\\
\text{parent of first left child on path from $v$ to root}, & \text{otherwise}.
\end{cases}
\]

In all three cases, the algorithm first computes $\phi(v)$, taking $O(n)$ time by Lemma~\ref{lemma:phicomputationlemma}. In the first case, it then descends down to the leftmost descendant of the right child, which takes $n-\phi(v)$ bubble operations, in $O(n)$ time. In the second case, the parent is reached by one operation (moving the last $1$ one position to the left if $v$ is a left child, and flipping the last $1$ if $v$ is a right child), taking $O(1)$ time. Finally, in the third case, we have up to depth of $v$ many steps of the latter kind, each taking constant time, so again in total $O(n)$ time. In all three cases, we get a total of $O(n)$ time before the next word is visited.
	 \end{proof}

\begin{figure}
\centering
\includegraphics[width=\textwidth]{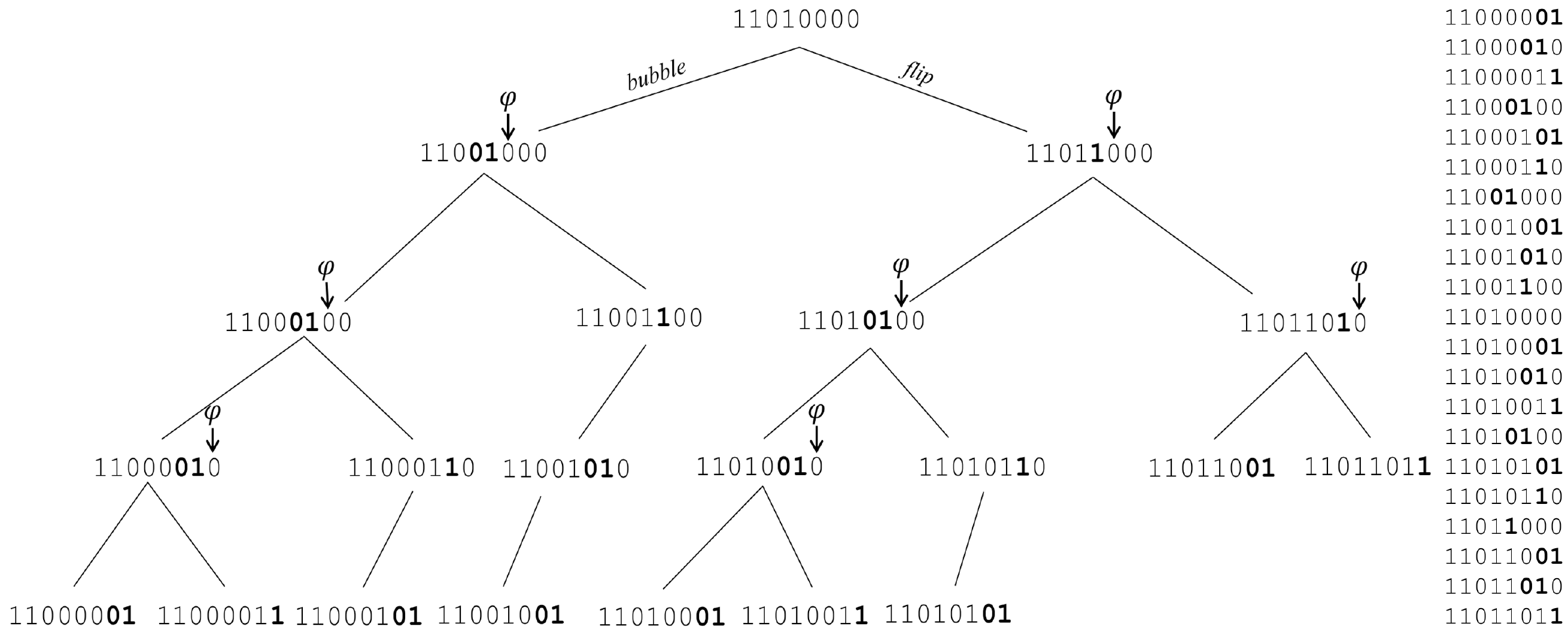}
  \caption{\label{fig:lexorder}The words in $\PNW(11010000)$ represented as a tree. If a node of the tree is word $w$, then its left child is $\bubble(w)$ and its right child is $\flip(w, \varphi(w)).$ In the tree, the position of $\varphi(w)$ is indicated, whenever $\phi(w)\leq n$; bubble operations (in the left child) resp.\ flip operations (in the right child) are highlighted in bold.  Algorithm \ref{algo:pnw_descent} generates these words by performing an in-order traversal of the tree. The corresponding list of words is given on the right.
}
\end{figure}

Now we are ready to present the full algorithm generating all prefix normal words of length $n$, see Algorithm \ref{algo:pnw_gen} ({\sc Bubble-Flip}). It first visits the two prefix normal words $0^n$ and $10^{n-1}$, and then generates recursively all words in ${\cal L}_n$ containing at least two $1$s, from the starting word $110^{n-2}$.

\IncMargin{1em}
\begin{algorithm}
	\DontPrintSemicolon
	{For a given $n$, generates all  prefix normal words of length $n$}\;
	\BlankLine

\nl	$w = 0^{n}$\;
\nl	$Visit()$\;
\nl	$w = 10^{n-1}$\;
\nl	$Visit()$\;
\nl	$w = 110^{n-2}$\;
\nl	{\sc Generate } $\PNW (w)$\;

\caption{{\sc Bubble-Flip}}\label{algo:pnw_gen}
\end{algorithm}
\DecMargin{1em}

\begin{theorem}\label{thm:algo}
The {\sc Bubble-Flip} algorithm generates all prefix normal words of length $n$, in lexicographic order, and in $O(n)$ time per word.
\end{theorem}

\begin{proof}
Recall that by Fact \ref{basicfacts}{\em (i)} every prefix normal word of length $n$, other than $0^n$, has $1$ as its first character. It is easy to see that there is only one prefix normal word of length $n$ with a single $1$, namely $10^{n-1}.$ Moreover, by Fact \ref{basicfacts}{\em (i)} and the definition of $\PNW$, the set of all prefix normal words of length $n$ with at least two $1$s coincides with $\PNW(110^{n-2})$. By Lemma~\ref{lemma:lex}, this set is generated by {\sc Generate $\PNW(110^{n-2})$} in lexicographic order and in $O(n)$ time per word. Noting that prepending $0^n$ and $10^{n-1}$ preserves the lexicographic order concludes the proof.  
\end{proof}

\subsection{Listing ${\cal L}_n$ as a combinatorial Gray code}

The algorithm {\sc Generate $\PNW(w)$} (Algorithm~\ref{algo:pnw_descent}) performs an in-order traversal of the nodes of the tree ${\cal T}(w)$. If instead we do a post-order traversal, we get a combinatorial Gray code of ${\cal L}_n$, as we will show next. First note that the change in the traversal order can be achieved by moving line 4 in Algorithm 2 to the end of the code, resulting in Algorithm~\ref{algo:pnw_descent2}.

\IncMargin{1em}
\begin{algorithm}
	\DontPrintSemicolon
	{Given a prefix normal word $w$ such that $|w|_1 > 1$, generate a combinatorial Gray code on $\PNW(w)$}\;
	\BlankLine

\nl	\If{$\RightmostOne(w) \neq n$}{
\nl		$w^{\prime} = \bubble(w)$\;
\nl		{\sc Generate2} $\PNW (w^{\prime})$\;
	}
\nl	$j = \varphi(w)$\;
\nl	\If{$j\leq n$}{
\nl		$w^{\prime\prime} = \flip(w,j)$\;
\nl		{\sc Generate2} $\PNW(w^{\prime\prime})$\;
	}
\nl	$Visit()$\;

	\caption{{\sc Generate2} $\PNW(w)$}\label{algo:pnw_descent2}
\end{algorithm}
\DecMargin{1em}

\begin{lemma}\label{lemma:postorder}
In a post-order traversal of ${\cal T}(w)$, two consecutive words have Hamming distance at most $3$.
\end{lemma}

\begin{proof}
Let $v$ be some node visited during the traversal of ${\cal T}(w)$. If $v$ is a \flip-node, then the next node in the listing will be its parent node $v'$. Since $v = \flip(v',\phi(v'))$,  $v'$ is at Hamming distance $1$ from $v$. Otherwise $v$ is a \bubble-node, i.e.\ $v = u010^k$ and its parent is $u10^{k+1}$ for some word $u$ and integer $k$. If $v$ has no right sibling, then the next node visited is its parent, at Hamming distance $2$ from $v$. Else the next node $v'$ is the leftmost descendant of $v$'s right sibling, i.e.\ $v' = u10^k1$, and the Hamming distance to $v$ is at most $3$. 
\end{proof}

\begin{exa} The words in $\PNW(11010000)$ (Fig.~\ref{fig:lexorder}) are listed by Algorithm~\ref{algo:pnw_descent2} as follows: $11000001$, $11000011$, $11000010$, $11000101$, $11000110$, $11000100$, $11001001,$ $11001010$, $11001100$, $11001000$, $11010001$, $11010011$, $11010010$, $11010101$, $11010110,$ $11010100$, $11011001$, $11011011$, $11011010$, $11011000$, $11010000$.
\end{exa}

\begin{theorem}\label{thm:gray-code}
The {\sc Bubble-Flip} algorithm using a post-order traversal produces a cyclic combinatorial Gray code on ${\cal L}_n$, generating each word in time $O(n)$.
\end{theorem}

\begin{proof}
By Lemma~\ref{lemma:postorder}, {\sc Generate2 $\PNW(110^{n-2})$} produces a combinatorial Gray code. By  visiting the two words $0^n$ and $10^{n-1}$ first, followed by {\sc Generate2 $\PNW(110^{n-2})$}, we get a combinatorial Gray code on all of ${\cal L}_n$. The last word in this code is the root $110^{n-2}$ and $d_H(110^{n-2},0^n) = 2 \leq 3$, thus this code is also cyclic.

Since only the order of the tree traversal changed w.r.t.\ the previous algorithm, it follows immediately that the algorithm visits ${\cal L}_n$ in amortized $O(n)$ time per word, since the overall running time is, as before, $O(n|{\cal L}|)$.

To see that the time to visit the next word is $O(n)$, we distinguish two cases according to the type of node. If $v$ is a flip-node, then the next node is its parent, taking $O(1)$ time to reach. If $v$ is a bubble-node, then we have to check whether it has a right sibling by computing $\phi(v')$, where $v'$ is the parent of $v$, in $O(n)$ time. If $\phi(v') >n$, then the next node is $v'$. If $\phi(v') \leq n$, then we have to reach the leftmost descendant of $\flip(v',\phi(v'))$, passing along the way only bubble-nodes. This takes $n-\phi(v')$ time, so altogether $O(n)$ time for the node $v$.

\end{proof}

\subsection{Prefix normal words with given critical prefix}\label{sec:critical_prefix}

Recall Definition~\ref{defi:criticalprefix}. It was conjectured in~\cite{BFLRS_CPM14} that the average length of the critical prefix taken over all prefix normal words is $O(\log n)$. Using the {\sc Bubble-Flip} algorithm, we can generate all prefix normal words with a given critical prefix $u$, which could prove useful in proving or disproving this conjecture. Moreover, if we succeed in counting prefix normal words with critical prefix $u=1^s0^t$, then this could lead to an enumeration of $|{\cal L}_n|$, another open problem on prefix normal words~\cite{BFLRS_TCS17}.

In the following lemma, we present  a characterization of prefix normal words of length $n$ with the same critical prefix $1^s0^t$ in terms of our generation algorithm. For $s\geq 1,t\geq 0$, let us denote by $\CritSet(s,t,n)$ the set of all prefix normal words of length $n$ and critical prefix $1^s0^t$. Note that there is only one prefix normal word whose critical prefix has $s=0$, namely $0^n$.

\begin{lemma}\label{lemma:critset}
Fix $s \geq 1$ and $t \geq 0$, and let $u = 1^s0^t.$ Then,
\begin{align*}
\CritSet(s,t,n)
= \begin{cases}
	\{u\} & \text{ if } s+t = n,  \\
	\{v\} \cup \PNW(\flip(v, \phi(v)), & \text{ if } s+t<n,\\
	\end{cases}
\end{align*}

where $v=u10^{n-(s+t+1)}$.
\end{lemma}

\begin{proof}
If $s+t = n$, then clearly $\CritSet(s,t, n) = \{u\}.$
Otherwise,
\begin{eqnarray*}
\CritSet(s,t,n) &=&
\{u10^{n-(s+t+1)}\} \cup \{u1\gamma \in {\cal L}_n  \mid  |\gamma|_1 >0\} \\
&=& \{v\} \cup \{u1\gamma \in {\cal L}_n  \mid \gamma_1,\ldots, \gamma_{\phi(v)-(s+t+2)}= 0, |\gamma|_1 >0\}\\
&=& \{v\} \cup \PNW(\flip(v, \varphi(v))),
\end{eqnarray*}
where the first equality holds by definition of critical prefix, the second by definition of $\phi(v)$,
and the third by definition of $\PNW$.
 \end{proof}

In Fig.~\ref{fig:critset_example}, we give a sketch of the placement of some of the sets with same critical prefix within ${\cal T}(110^{n-2})$, which, as the reader will recall, contains all prefix normal words of length $n$ except $0^n$ and $10^{n-1}$. The nodes in the tree are labelled with the corresponding generated word, and we have highlighted the subtrees corresponding to $\CritSet(1,1,n)$, $\CritSet(1,t,n)$, $\CritSet(s,1,n)$ and $\CritSet(s,t,n)$. Let us take a closer look at $\CritSet(s,t,n)$ for $s,t\geq 2$. The word $1^s0^t10^{n-(s+t+1)}$ is reached starting from the root $110^{n-2}$, traversing $s-1$ right branches (i.e.\ \flip-branches), passing through the word $1^s010^{n-(s+1)}$, and then traversing $t$ left branches (i.e.\ \bubble-branches). The set $\CritSet(s,t,n)$ is then equal to
the word $1^s0^t10^{n-(s+t+1)}$ together with its right subtree.

\begin{figure}
    \centering
    \includegraphics[width=\textwidth]{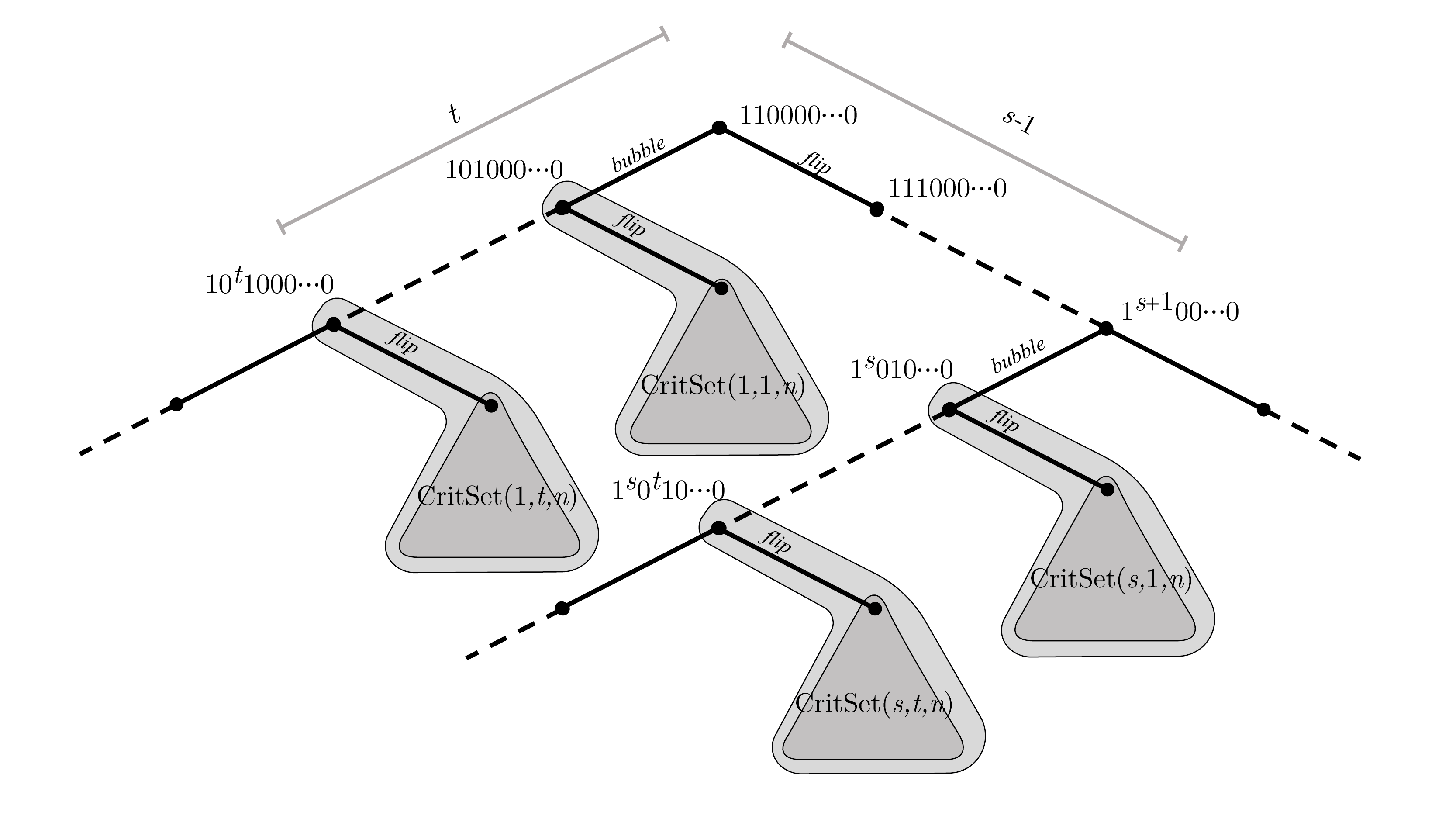}
      \caption{A sketch of the computation tree of Algorithm \ref{algo:pnw_descent} for the set $w=110^{n-2}$, highlighting the subtrees corresponding to sets of prefix normal words with the same critical prefix.  \label{fig:critset_example}
}
\end{figure}

Apart from revealing the recursive structure of sets of prefix normal words with the same critical prefix,
the {\sc Bubble-Flip} algorithm allows us to collect experimental data on the size of $\CritSet(s,t,n)$ for different values of $s,t,$ and $n$. We give some of these numbers, for $n=32$ and small values of $s$, see Table~\ref{tab:critsetsizes}. It was already known~\cite{BFLRS_CPM14} that, for $n\leq 50$, the average critical prefix length, taken over all $w\in {\cal L}_n$, is approximately $\log n$; with the new algorithm we are able to generate more precise data. In Fig.~\ref{fig:critprefixlength}, we plot the relative number of prefix normal words with a given critical prefix length, for lengths $n=16$ and $n=32$.

\begin{figure}
	\centering
	\includegraphics[width=\textwidth]{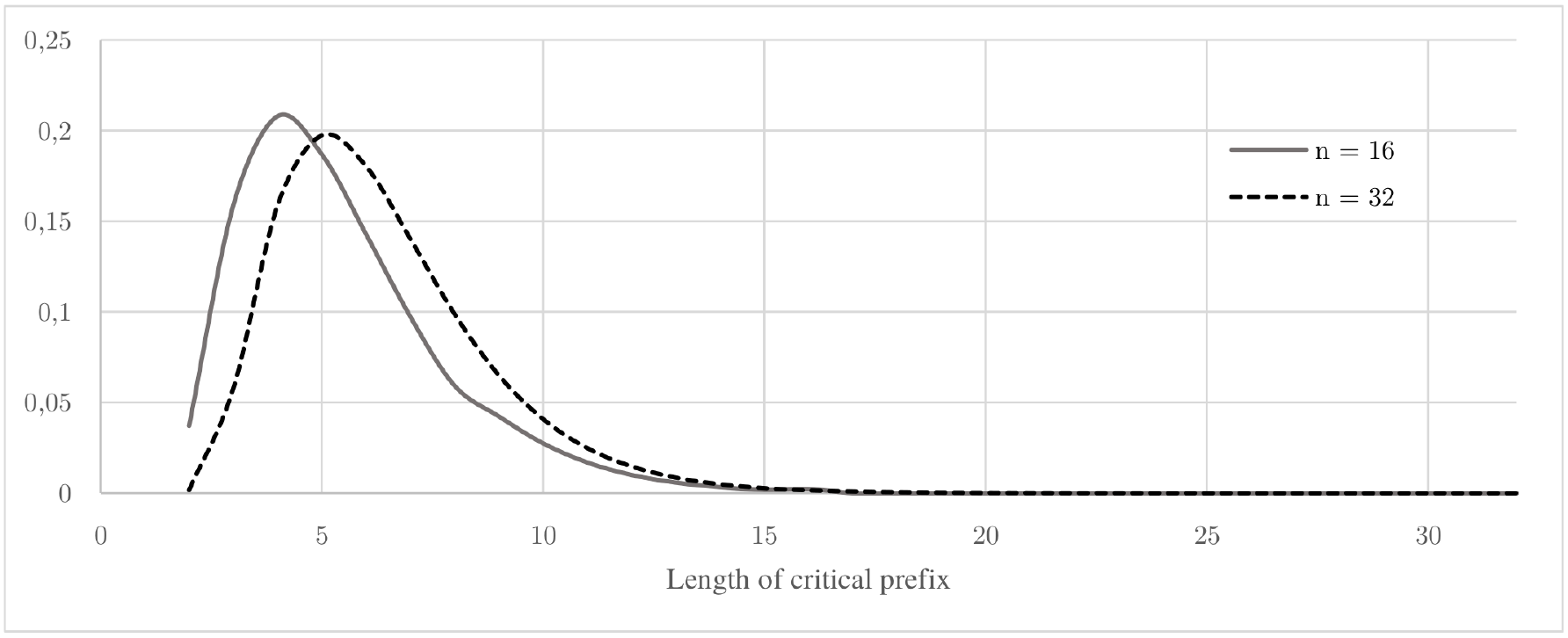}
	\caption{\label{fig:critprefixlength}%
	The frequency of prefix normal words with given critical prefix length, in percentage of the total number of prefix normal words of length $n$, for $n=16$ (solid) and $n=32$ (dashed). }
\end{figure}

\begin{table}[htbp]
	\centering
	\begin{tabularx}{\textwidth}{cr|>{\raggedleft\arraybackslash}p{0.7in}YYYYY}
		\toprule
		&       & \multicolumn{6}{c}{$t$} \\
		&       & 1     & 2     & 3     & 4     & 5     & 6      \\
		\midrule
		\multirow{7}[1]{*}{$s$}& \num{1} & \num{284663} & \num{14295} & \num{2226} & \num{597} & \num{220} & \num{100} \\
		& \num{2} & \num{9453217} & \num{979458} & \num{162336} & \num{38404} & \num{11679} & \num{4317} \\
		& \num{3} & \num{25025726} & \num{4907605} & \num{1103214} & \num{293913} & \num{91632} & \num{32459} \\
		& \num{4} & \num{27244624} & \num{7961078} & \num{2338632} & \num{732602} & \num{248717} & \num{91441} \\
		& \num{5} & \num{20423789} & \num{7521441} & \num{2677376} & \num{964483} & \num{360542} & \num{144460} \\
		& \num{6} & \num{12789981} & \num{5378726} & \num{2178190} & \num{874907} & \num{358717} & \num{151429} \\
		& \num{7} & \num{7270699} & \num{3301575} & \num{1454694} & \num{633310} & \num{276593} & \num{121726} \\
		\bottomrule
	\end{tabularx}%
\vspace{.05\linewidth}
	\begin{tabularx}{\textwidth}{cr|>{\raggedleft\arraybackslash}p{0.45in}>{\raggedleft\arraybackslash}p{0.43in}>{\raggedleft\arraybackslash}p{0.43in}>{\raggedleft\arraybackslash}p{0.35in}>{\raggedleft\arraybackslash}p{0.30in}YYYY}
		\toprule
		&       & \multicolumn{9}{c}{$t$} \\
		&       & 7 & 8 & 9 & 10 & 11 & 12 & 13 & 14 & 15   \\
		\midrule
		\multirow{7}[1]{*}{$s$}& \num{1} & \num{53} & \num{30} & \num{16} & \num{11} & \num{9} & \num{7} & \num{5} & \num{3} & \num{1} \\
		& \num{2} & \num{1788} & \num{813} & \num{451} & \num{276} & \num{161} & \num{90} & \num{47} & \num{16} & \num{15} \\
		& \num{3} & \num{12606} & \num{5815} & \num{2962} & \num{1475} & \num{723} & \num{346} & \num{121} & \num{106} & \num{92} \\
		& \num{4} & \num{37967} & \num{16994} & \num{7693} & \num{3507} & \num{1594} & \num{576} & \num{470} & \num{378} & \num{299} \\
		& \num{5} & \num{61139} & \num{26459} & \num{11658} & \num{5169} & \num{1941} & \num{1471} & \num{1093} & \num{794} & \num{562} \\
		& \num{6} & \num{65165} & \num{28543} & \num{12605} & \num{4944} & \num{3473} & \num{2380} & \num{1586} & \num{1024} & \num{638} \\
		& \num{7} & \num{54118} & \num{24188} & \num{9949} & \num{6476} & \num{4096} & \num{2510} & \num{1486} & \num{848} & \num{466} \\
		\bottomrule
	\end{tabularx}%
	\vspace{.05\linewidth}
	\begin{tabularx}{\textwidth}{cr|>{\raggedleft\arraybackslash}p{0.25in}>{\raggedleft\arraybackslash}p{0.25in}>{\raggedleft\arraybackslash}p{0.25in}>{\raggedleft\arraybackslash}p{0.15in}>{\raggedleft\arraybackslash}p{0.15in}>{\raggedleft\arraybackslash}p{0.15in}>{\raggedleft\arraybackslash}p{0.15in}>{\raggedleft\arraybackslash}p{0.15in}>{\raggedleft\arraybackslash}p{0.15in}YYYYYYYY}
		\toprule
		&       & \multicolumn{17}{c}{$t$} \\
		&       & 16 & 17 & 18 & 19 & 20 & 21 & 22 & 23 & 24 & 25 & 26 & 27 & 28 & 29 & 30 & 31 & 32   \\
		\midrule
		\multirow{7}[1]{*}{$s$} & \num{1} & \num{1} & \num{1} & \num{1} & \num{1} & \num{1} & \num{1} & \num{1} & \num{1} & \num{1} & \num{1} & \num{1} & \num{1} & \num{1} & \num{1} & \num{1} & \num{1} & \num{0} \\
		& \num{2} & \num{14} & \num{13} & \num{12} & \num{11} & \num{10} & \num{9} & \num{8} & \num{7} & \num{6} & \num{5} & \num{4} & \num{3} & \num{2} & \num{1} & \num{1} & \num{0} & \num{0} \\
		& \num{3} & \num{79} & \num{67} & \num{56} & \num{46} & \num{37} & \num{29} & \num{22} & \num{16} & \num{11} & \num{7} & \num{4} & \num{2} & \num{1} & \num{1} & \num{0} & \num{0} & \num{0} \\
		& \num{4} & \num{232} & \num{176} & \num{130} & \num{93} & \num{64} & \num{42} & \num{26} & \num{15} & \num{8} & \num{4} & \num{2} & \num{1} & \num{1} & \num{0} & \num{0} & \num{0} & \num{0} \\
		& \num{5} & \num{386} & \num{256} & \num{163} & \num{99} & \num{57} & \num{31} & \num{16} & \num{8} & \num{4} & \num{2} & \num{1} & \num{1} & \num{0} & \num{0} & \num{0} & \num{0} & \num{0} \\
		& \num{6} & \num{382} & \num{219} & \num{120} & \num{63} & \num{32} & \num{16} & \num{8} & \num{4} & \num{2} & \num{1} & \num{1} & \num{0} & \num{0} & \num{0} & \num{0} & \num{0} & \num{0} \\
		& \num{7} & \num{247} & \num{127} & \num{64} & \num{32} & \num{16} & \num{8} & \num{4} & \num{2} & \num{1} & \num{1} & \num{0} & \num{0} & \num{0} & \num{0} & \num{0} & \num{0} & \num{0} \\
		   \bottomrule
	\end{tabularx}%
	\caption{The size of $\CritSet(s,t,n)$ for $n=32$, $s=1,\ldots,7$ and $t = 1, \ldots, 32$\label{tab:critsetsizes}}%
\end{table}%

\subsection{Practical improvements of the algorithm}

The  running time of the algorithm is dominated by the time spent at each node  for computing the value of $\phi$, which, in general, takes
time linear in $n$, the length of the words. Therefore the overall generation of ${\cal L}_n$ takes $O(n |{\cal L}_n|)$ time.
One way of improving the running time of the overall generation would be to achieve faster amortized computation of $\phi$ by exploiting
the relationship between $\phi(w)$ and $\phi(w')$ for words $w$ and $w'$ generated at close nodes of the recursion tree. Next we
present two attempts in this direction. We show two cases where the $\phi(w)$ can be computed in sublinear time.
This implies a faster generation algorithm, absolutely, however,
since the number of nodes falling in such cases is only $o(|{\cal L}_n|)$
we do not achieve any significant asymptotic improvement on the overall generation.

\medskip

The first practical improvement can be obtained from the following lemma. It shows that given a node $w$ of the generation tree, for all nodes $w'$ in the subtree rooted in $w$, which are reachable from $w$ by traversing only  flip-branches,
the value $\phi(w')$ can be computed in time $O(\RightmostOne(w))$.
Note that on such a {\em rightward-path} words have a strictly increasing number of $1$s.
Therefore, the result of the lemma provides
a strict improvement on the original estimate that for each word $w'$ in such
rightward-path the computation of  $\phi(w')$ requires $\Theta(\RightmostOne(w')).$
This gives an improvement for nodes along the right branches of the tree only; the improvement gets better as we move further down a right path.

\begin{lemma}\label{lemma:lengthw-new}
Let $w \in {\cal L}_n$ and let
$$v^{(j)} =
\begin{cases}
w & j= 0\cr
\flip(v^{(j-1)}, \phi(v^{(j-1)})) & j > 0
\end{cases}
$$
i.e., $v^{(j)}$ is the word produced by applying $j$ times the $\flip$ operation starting from $w$.
For each $j \geq 0$ and $k \geq 1$, we have that $v = \flip(v^{(j)}, r(v^{(j)}) + k)$ is in ${\cal L}_n$  if and only if for all $t=1, \dots, r(w)$ it holds that
$|v_{r(v^{(j)})+k-t+1} \cdots v_{r(v^{(j)})+k}|_1$ $\leq |w_1 \cdots w_t|_1,$ i.e., the suffix of $v_1 \cdots v_{r(v^{(i)})+k}$ of length $t$
satisfies the prefix normal condition.
\end{lemma}

\begin{proof}
Assume otherwise and let $j$ and $k$ be the smallest integers such that $v = \flip(v^{(j)}, r(v^{(j)}) + k)$ is a counterexample---we first choose the smallest
$j$ such that there is a $k$ and then among all such $k$'s we choose the smallest, given the choice of $j$.

Let  $n_0 = r(v^{(j)}) + k.$
We write $P(i)$ for $P_v(i)$, and denote by $S(i)$ the number of $1$s in the $i$-length suffix of $v_1 \cdots v_{n_0}$.
Let $r = r(w)$. By assumption, $S(t)\leq P(t)$ for all $t\leq r$, but there is an $m > r$ such that $S(m) > P(m)$.  Choose this $m$ minimal.
By definition, using the properties of the $\phi$ function, we have that $v^{(j)} \in {\cal L}_n$. Moreover, by the minimality
of the choice of $j$ and $k$,  it holds that  $v_{n_0-m+1} \cdots v_{n_0-1}$ satisfies the prefix normal condition, i.e.\
$|v_{n_0-m+1} \cdots v_{n_0-1}|_1 \leq P(m-1)$. Therefore, it must hold that $P(m-1) = P(m)$, hence $v_m=0$.
Since $m>r$ and $v_m=0$, there must be $0 \leq j' \leq  j$ such that $\phi(v^{(j')}) < n_0-m < \phi(v^{(j'+1)})$, i.e., the
flip operation that produces $v^{(j'+1)}$ has to be done on a position following $n_0-m$.
 This means that for some $t'<m$, $|v_{m-t'+1}\cdots v_m|_1 = P(t'),$ otherwise we would have $v_m=1$. Let $m' = m-t'$. Thus we have $P(m) = P(m')+P(t')$. On the other hand, $S(m) = S(m') + |v_{n_0-m+1} \cdots v_{n_0-m+t'}|_1 \leq P(m') + P(t') = P(m)$, where the inequality holds by the minimality of $m$ and of $n_0$, respectively. But this is a contradiction to our assumption that $S(m)>P(m)$.

\end{proof}

Second, we show how to derive $\phi(v')$ for a \bubble-node $v'$ from $\phi(v)$, where $v$ is the parent of $v'$. This gives an improvement (from linear to constant) for all nodes of the form $\bubble^*(v)$ of some node $v$, spreading out the cost of computing $\phi(v)$ for $v$ over all \bubble-descendants of $v$. Note that this covers the case of Observation~\ref{obs:tree_properties}, part 3, which tells us that we can skip the computation of $\phi(v)$ if the parent of $v$ does not have a \flip-child.

\begin{lemma}
	Let $w$ be a  prefix normal word $w$ of length $n$ with $|w|_1 \geq 2$ and $r = \RightmostOne(w) \neq n$. Then
\begin{equation}\label{bubbleflip-eq}
\varphi(\bubble(w)) = \begin{cases}
	\min\{n+1, \varphi(w)+2\} & \text{if }|w|_1 = 2 ,\\
	\varphi(w) & \text{if } |w_1 \cdots w_{\varphi(w)-r}|_1 > 1, \\
	\min\{n+1, \varphi(w)+1\} & \text{otherwise}. \\
	\end{cases}
\end{equation}

In particular, $\phi(\bubble(w))$ can be computed in constant time, given $\phi(w)$.
\end{lemma}

\begin{proof}
An immediate observation is that $\varphi(w) \leq \varphi(\bubble(w)).$ Therefore, if $\varphi(w) = n+1$ the claim holds trivially. \medskip

\noindent
{\em Case 1.} $|w|_1 = 2.$ Then we can write $w$ as $w = 10^{r-2}10^{n-r}$ and $\bubble(w)  = 10^{r-1}10^{n-r-1}$.
It is then easy to see that  we have  $\varphi(w) = \varphi(10^{r-2}10^{n-r}) = \min\{n+1, r+t+1\}$
and $\varphi(\bubble(w)) = \varphi(10^{r-1}10^{n-r-1}) = \min\{n+1, (r+1) + (t+1) + 1\} = \min\{n+1, \varphi(w) + 2\}$, as desired. Since $w \in {\cal L}_n,$ we have $w_1 = 1.$

\medskip

\noindent {\em Case 2.} $|w_1\cdots w_{\varphi(w) - r}|_1 > 1.$
First of all, let us observe that we have $|w|_1 > 2$. For otherwise, if $|w|_1 = 2$, the analysis of the previous case
implies that $\varphi(w) - r = r-1$ hence $|w_1\cdots w_{\varphi(w) - r}|_1 = 1,$ contradicting the standing hypothesis.
From $|w_1\cdots w_{\varphi(w) - r}|_1 > 1$, it follows that $\varphi(w) > r+1.$ Moreover,  we have $\varphi(w) < 2r,$
since $w_1 \cdots w_{r-1} 1 0^{r-2} 1 \in {\cal L}_n$ (by Fact \ref{basicfacts} {\em (iv)}).

Now let $w' = \flip(w, \varphi(w))$ and $w'' = \flip(\bubble(w), \varphi(w)),$ i.e.,
\begin{align*}
w' &= w_1 \cdots w_{r-1} 1 0^{\varphi(w)-r-1} 1 0^{n-\varphi(w)}, \\
w'' &= w_1 \cdots w_{r-1} 0 1 0^{\varphi(w)-r-2} 1 0^{n-\varphi(w)}.
\end{align*}

By the definition of $\varphi$ we have $w' \in {\cal L}_n$. Moreover,  it holds that $\bubble(w) = w_1 \cdots w_{r-1}01 0^{n-r-1} \in {\cal L}_n.$
For proving the claim, it is enough to show that  $w'' \in {\cal L}_n.$

Let $S_{w''}(i) = |w''_{\varphi(w)-i+1} \cdots w''_{\varphi(w)}|_1$ and
$S_{w'}(i) = |w'_{\varphi(w)-i+1} \cdots w'_{\varphi(w)}|_1.$
It is not hard to see that for each $i \not \in \{r,  \varphi(w) - r\}$,
it holds that  $S_{w''}(i) = S_{w'}(i) \leq P_{w'}(i) = P_{w''}(i),$ where the inequality follows from the prefix normality
of $w'$.
Moreover,
for $i = \varphi(w)-r$, we have $S_{w''}(\varphi(w)-r) = 2$ and since $\varphi(w) - r < r$, we also have
$P_{w''}(\varphi(w)-r) =  P_{w'}(\varphi(w)-r) = P_{w}(\varphi(w)-r) > 1$ (by the standing hypothesis).
Finally, for $i = r$, using again $\varphi(w)-r < r$, it follows that $S_{w''}(r) = S_{w'}(r) \leq P_{w'}(r) = P_{w''}(r).$
In conclusion, we have $S_{w'}(i) \leq P_{w'}(i)$ for each $i = 1, \dots, \varphi(w),$ hence,
by  Fact \ref{basicfacts} {\em (iv)},  the word $w_1\cdots w_{r-1} 0 1 0^{\varphi(w)-r-2} 1 \in {\cal L}$ and
by Fact \ref{basicfacts} {\em (iii)}, $w'' \in {\cal L}_n,$ which concludes the proof of this case.

\medskip

\noindent {\em Case 3.} $|w_1\cdots w_{\varphi(w) - r}|_1 = 1$ and $|w|_1 > 2$. Proceeding
as in the previous case, we have that
$S_{w'}(\varphi(w)-r) = 2 > P_{w}(\varphi(w)-r) =  P_{w'}(\varphi(w)-r),$ which implies that
$w' \not \in {\cal L}_n,$ hence $\varphi(\bubble(w)) \geq \varphi(w)+1.$
Let
$$w''' = w_1 \cdots w_{r-1} 0 1 0^{\varphi(w)-r-1} 1 0^{n-\varphi(w)-1} = \flip(\bubble(w), \varphi(w)+1).$$

It is enough to show that $w''' \in {\cal L}_n.$
Let us redefine $S_{w'''}(i) = |w'''_{\varphi(w)-i+2} \cdots w'''_{\varphi(w)+1}|_1$ and
$S_{w'}(i) = |w'_{\varphi(w)-i+1} \cdots w'_{\varphi(w)}|_1.$
It is not hard to see that for each $i \in \{1, \dots, \varphi(w)\}$,
it holds that  $S_{w'''}(i) \leq S_{w'}(i)$. Moreover,
for each $i \in \{1, \dots, \varphi(w)-1\} \setminus \{r\}$, we have $P_{w'}(i) = P_{w'''}(i).$ Thus,
for each $i \in \{1, \dots, \varphi(w)-1\} \setminus\{r\}$,
it holds that  $S_{w'''}(i) \leq S_{w'}(i) \leq P_{w'}(i) = P_{w'''}(i),$
where the second inequality follows from the prefix normality of $w'.$

For $i = \varphi(w)$, using $w'''_1 = 1 = w'''_{\varphi(w)+1}$
we have $S_{w'''}(\varphi(w))  =  |w|_1 = P_{w'''}(\varphi(w)).$

For $i=r$, we have $\varphi(w)+2- r \leq r+1$.
If $\varphi(w)+2- r = r+1,$ i.e., $\varphi(w) = 2r-1$ then
$S_{w'''}(r)  = 2 = |w'_{r} \cdots w'_{\varphi(w)}|_1 \leq P_{w'}(r).$ Since $P_{w'}(r) = |w|_1 \geq 3,$ we have
$P_{w'''}(r) = P_{w'}(r)-1 \geq 2 = S_{w'''}(r).$

If $\varphi(w)+2- r \leq r,$
then
\begin{eqnarray*}
P_{w'''}(r) - S_{w'''}(r)  &=& |w'''_1\cdots w'''_r|_1 - |w'''_{\varphi(w)+2- r} \dots w'''_{\varphi(w)+1}|_1 \\
&=&
 |w'''_1\cdots w'''_{\varphi(w)+1- r}|_1 - |w'''_{r+1} \cdots w'''_{\varphi(w)+1}|_1 \\
 &=& P_{w'''}(\varphi(w)+1- r) - S_{w'''}(\varphi(w)+1- r) \geq 0,
 \end{eqnarray*}
 where the middle equality follows by removing from the two words the common intersection, and the last inequality comes from the
 previous subcase, as $\varphi(w)+1- r \in \{1, \dots ,\varphi(w)-1\} \setminus \{r\}.$

In conclusion, we have $S_{w'''}(i) \leq P_{w'''}(i)$ for each $i = 1, \dots, \varphi(w)+1,$ hence,
by  Fact \ref{basicfacts} {\em (iv)}  the word $w_1\cdots w_{r-1} 0 1 0^{\varphi(w)-r-1} 1 \in {\cal L}$ and
by Fact \ref{basicfacts} {\em (iii)} $w''' \in {\cal L}_n,$ which concludes the proof of this case.
The proof of (\ref{bubbleflip-eq}) is complete.

We now argue that $\phi(\bubble(w))$ can be computed in constant time.
Our result says that knowing $\RightmostOne(w)$ and the position of the second leftmost $1$ in $w$, then $\phi(\bubble(w))$ can be
computed applying (\ref{bubbleflip-eq}), i.e., in $O(1)$ time. In fact, the condition
$|w_1 \cdots w_{\varphi(w)-r}|_1 > 1$ is equivalent to checking that the second leftmost $1$ of $w$ is in a position
not larger than  $\varphi(w)-\RightmostOne(w).$
It is not hard to see that $\RightmostOne(w)$
and the position of the second leftmost $1$ of
$w$ can be computed and maintained for each node on the generation tree
without increasing the computation by more than a constant
amount of time per node.

\end{proof}

We provide the following examples to illustrate the two improvements.

\begin{example}
For the first improvement, consider the word $w = 11001010^{n-7}$ with $n$ some large number. Let $w^{(1)}, w^{(2)}, \dots, w^{(i)}$ be the words generated on the right path rooted at $w$, i.e., $w^{(1)}$ is the flip-child of $w$,
$w^{(2)}$ is the flip-child of $w^{(1)}$ and so on.

It is not hard to see that
$w^{(1)} = 1100101010^{n-9}$, $w^{(2)} = 110010101010^{n-11},$ and in general
$w^{(i)} = 1100101(01)^i 0^{n-7-2i}$ for any $i = 1,2,\dots, \frac{n-7}2.$

What Lemma 9 guarantees is that,  for $i = 1, \dots, \frac{n-7}2$,  $w^{(i)}= \flip(w^{(i-1)}, \phi(w^{(i-1)}))$
can be computed in time $\Theta(r(w))$ rather than $\Theta(r(w^{(i-1)})).$
Therefore, in total, to generate them all,  we need
 $\Theta(r(w)\cdot n)$. Without applying Lemma 9, i.e., computing
 $w^{(i)} = \flip(w^{(i-1)}, \phi(w^{(i-1)}))$ using Algorithm 1, in time $r(w^{(i-1)}) = 7+2(i-1)$,
we would need in total  $\Theta(n^2)$ time.
\end{example}

\begin{example}
    For the second improvement,  consider the word $w = 100100000000$, for which it holds that $|w|_1 = 2$.
    We have that $\varphi(w) = 7$, and indeed, the word $\bubble(w) = 100010000000$ has $\varphi(\bubble(w)) = 9$.  As an example for a word $w$ with $|w_1 \cdots w_{\varphi(w)-r}|_1 > 1$, consider the word $w = 110001010000$. We have $\varphi(w) = 11$ and also for the word $\bubble(w) =  110001001000$, we have $\varphi(\bubble(w)) = 11$.
Finally, consider the word $w = 101001001000$. We have $\varphi(w) = 11$, and since $|10|_1 \leq 1$, it holds that $\bubble(w) = 101001000100$ and $\varphi(\bubble(w)) = 12$.
\end{example}

\section{On finite and infinite prefix normal words}\label{sec:theory}

In this section, we study infinite prefix normal words. We focus on infinite extensions
of finite prefix normal words which satisfy the prefix normal condition at every finite point and which are in a certain sense densest among all possible infinite extensions of the starting word. We show that words in
this class are ultimately periodic, and we are able to determine both the size and the density of the period and to upper bound the starting point of the periodic behaviour.

\subsection{Definitions}

An infinite binary word is a function $v: \IN \to \{0,1\}$ (where $\IN$ denotes the set of natural numbers not including $0$). The set of all infinite binary words is denoted $\{0,1\}^{\omega}$. As with finite words, we refer to the $i$'th character of $v$ by $v_i$, to the factor spanning positions $i$ through $j$ by $v_i\cdots v_j$, and to the prefix of length $i$ by $\pref_i(v)$. As before, $P(i) = P_v(i)$ denotes the number of $1$s in the prefix of length $i$. Given a finite word $u$, $u^{\omega}$ denotes the infinite word $uuu\cdots$.
An infinite word $v$ is called {\em ultimately periodic} if there exist two integers $p,i_0\geq 1$ such that $v_{i+p} = v_i$ for all $i\geq i_0$, or equivalently, if it can be written as $v=zu^{\omega}$ for some finite words $z,u$. The word $v$ is called {\em periodic} if it is ultimately periodic with $i_0=1$, or equivalently, if there exists a finite word $u$ such that $v = u^{\omega}$. If $v=zu^{\omega}$, then we refer to $u$ as a period of $v$.%

\begin{definition}[Minimum density, minimum density prefix] %
Let $w\in \{0,1\}^*\cup \{0,1\}^{\omega}$. Denote by $D(i) = D_w(i) = P_w(i)/i$, the {\em density} of the prefix of length $i$. Define the {\em minimum density of $w$} as $\delta(w) = \inf \{ D(i) \mid 1\leq i \}$. If this infimum is attained somewhere, then we also define
\[
\iota(w) = \min \{ j \mid \forall i: D(j) \leq D(i) \}, \quad \text{ and } \quad \kappa(w) = P_w(\iota(w)).
\]

We refer to $\pref_{\iota(w)}(w)$ as the {\em minimum-density prefix}, the shortest prefix with density $\delta(w)$.
Note that $\iota(w)$ is always defined for finite words, while for infinite words, a prefix which attains the infimum may or may not exist.
\end{definition}

\begin{exa}
For $w = 110100101001$ and $u = 110100101010$ we have $\delta(w)=5/11, \iota(w) = 11, \kappa(w)=5,$ and $\delta(u) = 1/2, \iota(u) = 6, \kappa(u)=3$. For the infinite words $v = (10)^{\omega}$ and $v' = 1(10)^{\omega}$, we have $\delta(v) = \delta(v') = 1/2$, and $\iota(v)=2, \kappa(v)=1$, while $\iota(v')$ is undefined, since no prefix attains density $1/2$.
\end{exa}

For a prefix normal word $u$, every factor of the infinite word $u0^{\omega}$ respects the prefix normal condition. This  leads to the definition of infinite prefix normal words.

\begin{definition}[Infinite prefix normal words]
An infinite binary word $v$ is called {\em prefix normal} if, for every factor $u$ of $v$, $|u|_1 \leq P_v(|u|)$.
\end{definition}

Clearly, as for finite words, it holds that an infinite word is prefix normal if and only if all its prefixes are prefix normal. Therefore, the existence of infinite prefix normal words can also be derived from K\"onig's Lemma (see~\cite{Lothaire3}), which states that the existence of an infinite prefix-closed set of finite words implies the existence of an infinite word which has all its prefixes in the set.

We now define an operation on finite prefix normal words which is similar to the $\flip$ operation from Sec.~\ref{sec:algorithm}: it takes a prefix normal word $w$ ending in a $1$ and {\em extends} it by a run of $0$s followed by a new $1$, in such a way that this $1$ is placed in the first possible position without violating prefix normality.

\begin{definition}[Operation $\flipext$]
Let $w\in {\cal L} \cap \{0,1\}^*1$. Define $\flipext(w)$ as the finite word $w0^k1$, where $k=\min \{ j \mid w0^j1 \in {\cal L}\}$. We further define the infinite word $v = \flipext^{\omega}(w) = \lim_{i\to \infty} \flipext^{(i)}(w)$.
\end{definition}

For a prefix normal word $w$, the word $w0^{|w|}1$ is always prefix normal, so the operation \flipext\ is well-defined.
Let $w \in {\cal L}$ and $r=\RightmostOne(w)<|w|$. Then $\flipext(\pref_{r}(w))$ is a prefix of $\flip(w, \varphi(w))$ if and only if $\phi(w) \leq |w|$, in particular, $\flip(w, \varphi(w)) = \flipext(\pref_{r}(w))\cdot 0^{|w|-\phi(w)}.$

\begin{definition}[Iota-factorization]
Let $w$ be a finite binary word, or an infinite binary word such that $\iota = \iota(w)$ exists. The iota-factorization of $w$ is the factorization of $w$ into $\iota$-length factors, i.e.\ the representation of $w$ in the form
\begin{align*}
w & =  u_1 u_2\cdots u_{r}v, \\
&\text{ where $r = \floor{|w|/\iota}$, $|u_i| = \iota$  for $i=1,\ldots,r$, and $|v| < \iota$}, &
 \text{for $w$ finite, and }\\
w & =  u_1 u_2 \cdots, \text{ where  $|u_i| = \iota$ for all $i$}, & \text{for $w$ infinite.}
\end{align*}

\end{definition}

\subsection{Flip extensions and ultimate periodicity}

\begin{lemma}~\label{lemma:iota-factorization}
Let $w$ be a finite or infinite prefix normal word, such that $\iota = \iota(w)$ exists. Let $w=u_1u_2\cdots$ be the iota-factorization of $w$. Then for all $i$, $|u_i|_1 = \kappa(w)$.
\end{lemma}

\begin{proof}
Since $w$ is prefix normal, $|u_i| \leq \kappa = \kappa(w)$. On the other hand, assume there is an $i_0$ for which $|u_{i_0}|_1 < \kappa$. Then the prefix $u_1u_2\cdots u_{i_0}$ has fewer than $i_0\kappa$ many $1$s, and thus density less than $i_0\kappa / i_0 \iota = \kappa/\iota = D(\iota)$, in contradiction to the definition of $\iota$. 
\end{proof}

The next lemma states that the iota-factorization of a word $w$ constitutes a non-increasing sequence w.r.t.\ lexicographic order, as long as $w$ fulfils a weaker condition than prefix normality, namely that factors of length $\iota(w)$ obey the prefix normal condition. That this does not imply prefix normality can be seen on the example $(1110010)^{\omega}$, which is not prefix normal.

\begin{lemma} \label{lemma:lexicographic-order-of-the-factors}
Let $w$ be a finite or infinite binary word, such that $\iota = \iota(w)$ exists. Let $w=u_1u_2\cdots$ be the iota-factorization of $w$. If for every $i$, $|u_i|_1 = \kappa = \kappa(w)$, and every factor $u$ of length $\iota$ fulfils the prefix normal condition, then for all $i$, $u_i \geq_{\lex} u_{i+1}$.
\end{lemma}

\begin{proof} Let us write $u_i = u_{i,1} \cdots u_{i,\iota}$. Let $a(i,j) = |u_{i,1} \cdots u_{i,j}|_1$ denote the number of $1$s in the $j$-length prefix of $u_i$, and
$b(i,j) = |u_{i,j+1} \cdots u_{i,\iota}|_1$ the number of $1$s in the suffix of length $\iota-j$. By Lemma~\ref{lemma:iota-factorization}, we have that $a(i,j) + b(i,j) = \kappa$. On the other hand, $b(i,j) + a(i+1,j) \leq \kappa$, since all $\iota$-length factors satisfy the prefix normal condition. Thus, for all $i$: $a(i,j) \geq a(i+1,j)$.

If $u_i \neq u_{i+1}$, let $h = \min \{j\mid j = 1, \ldots, \iota : a(i,j)> a(i+1,j) \}$.
Thus, for every $j < h$, we have $ u_{i,j} = u_{i+1,j}$ and  $u_{i,h} =1, u_{i+1,h}=0$, implying $u_i \geq_{\lex} u_{i+1}$.  

\end{proof}

\begin{corollary}\label{coro:lexicographic-order-of-the-factors}
Let $w$ be a finite or infinite prefix normal word, such that $\iota = \iota(w)$ exists. Then for all $i$, $u_i \geq_{\lex} u_{i+1}$, where $u_i$ is the $i$'th factor in the iota-factorization of $w$.
\end{corollary}

We now prove  that the \flipext\ operation leaves the minimum density invariant.  This means that among all infinite prefix normal extensions of a word $w\in {\cal L}$, the word $v=\flipext^{\omega}(w)$ has the highest minimum density.

\begin{lemma}\label{lemma:invariant-iota}
Let $w \in {\cal L}$ such that $w_n=1$, and let $v \in \flipext^*(w) \cup \{\flipext^{\omega}(w)\}$. Then $\delta(v) = \delta(w)$, and as a consequence, $\iota(v) = \iota(w)$ and $\kappa(v)=\kappa(w)$.
\end{lemma}

\begin{proof}
Assume otherwise. Then there exists a minimal index $i$ such that $D_v(i) < \delta(w)=:\delta$. Clearly, $i> |w|$, by definition of $\delta$. Since $i$ is minimal, it follows that $D_v(i-1)\geq \delta$, which implies  $v_i=0$. Since $i> |w|$, there was an iteration of \flipext, say the $j$'th iteration, which produced the extension containing position $i$, i.e.\ $|\flipext^{(j-1)}(w)| < i < |\flipext^{(j)}(w)|$. Since $v_i=0$, this implies that there is an $m$ such that the factor $v_{i-m+1}\cdots v_{i-1}1$ would have violated the prefix normal condition, i.e.\ $|v_{i-m+1}\cdots v_{i-1}1|_1 > P_v(m)$. This implies $|v_{i-m+1}\cdots v_{i-1}0|_1 = P_v(m)$ (because $v$ is prefix normal). Now consider the prefix $\pref_i(v) = v_1\cdots v_i$, and let us write $i = i' + m$.  Since $i$ was chosen minimal, we have that $D_v(i'), D_v(m)\geq \delta$. Since $D_v(i') =\frac{P_v(i')}{i'}, D_v(m) = \frac{P_v(m)}{m}$, this implies
\[ D_v(i) = \frac{P_v(i)}{i} = \frac{P_v(i') + P_v(m)}{i' + m} \geq \delta,
\]
in contradiction to the assumption. 
\end{proof}

\begin{theorem}\label{thm:ult-periodic}
Let $w\in{\cal L}$ and $v=\flipext^{\omega}(w)$. Then $v$ is ultimately periodic. In particular, $v$ can be written as $v = ux^{\omega}$, where $|x| = \iota(w)$ and $|x|_1 = \kappa(w).$
\end{theorem}

\begin{proof}
By Lemma~\ref{lemma:invariant-iota}, $\iota(v) = \iota(w)$, and by Lemma~\ref{lemma:iota-factorization}, in the iota-factorization of $w$, all factors $u_i$ have $\kappa(w)$ $1$s. Moreover,
 by Corollary~\ref{coro:lexicographic-order-of-the-factors}, the factors $u_i$ constitute a lexicographically non-increasing sequence. Since all $u_i$ have the same length $\iota(w)$, and there are finitely many binary words of length $\iota(w)$, the claim follows.

\end{proof}

We can further show that the period $x$ from the previous theorem is prefix normal, as long as it starts in a position which is congruent 1 modulo $\iota$, in other words, if it is one of the factors in the iota-factorization of $v$.

\begin{lemma}\label{lemma:xpn} Let $w\in{\cal L}$ and $v=\flipext^{\omega}(w) = ux^{\omega}$ such that $x$ is the $k'th$ factor in the iota-factorization of $v$, for some $k\geq 1$. Then $x$ is prefix normal.
\end{lemma}

\begin{proof} First note that if $v = x^{\omega}$, then $x$ is prefix normal by the prefix normality of $v$.
Else, assume for a contradiction that $x$ is not prefix normal. Let $\alpha$ be a factor of $x$ of minimal length s.t.\ $|\alpha|_1 > |\beta|_1$, where $\beta$ is the prefix of $x$ of length $|\alpha|$. Then $\beta$ and $\alpha$ are disjoint due to the minimality assumption. In other words, there is a (possibly empty) word $\gamma$ s.t.\ $\beta\gamma\alpha$ is a prefix of $x$.

Since $x$ is a $\iota$-factor of $v$, therefore the prefix of $v$ before $x$ has length $t\iota$ for some $t\geq 1$. Let $x = \beta\gamma\alpha\nu$, and write $x'$ for the rotation $\nu\beta\gamma\alpha$ of $x$. Now consider the word $s = \gamma\alpha (x')^t$, which has length $|\gamma| + |\alpha| + t\iota$. By Theorem~\ref{thm:ult-periodic}, $|x|_1=\kappa$, and since $x'$ is a rotation of $x$, also $|x'|_1=\kappa$. Therefore, for the factor $s$ of $v$ it holds that $|s|_1 = |\gamma|_1 + |\alpha| _1+ t\kappa > |\gamma|_1 + |\beta| _1+ t\kappa = P_v(|s|)$, in contradiction to $v\in {\cal L}$. 
\end{proof}

Next we show that for a word $v\in \flipext^*(w)$, in order to check the prefix normality of an extension of $v$,
it  is enough to verify that the suffixes up to length $|w|$ satisfy the prefix normal condition.

\begin{lemma}\label{lemma:lengthw}
Let $w$ be prefix normal and $v'\in \flipext^*(w)$. Then for all $k\geq 0$ and $v=v'0^k1$, $v\in{\cal L}$ if and only if for all $1\leq j \leq |w|$, the suffixes of $v$ of length $j$ satisfy the prefix normal condition.
\end{lemma}

\begin{proof}
Directly from Lemma \ref{lemma:lengthw-new}.
\end{proof}

By Theorem~\ref{thm:ult-periodic}, we know that $v=\flipext^{\omega}(w)$ has the form $v=ux^{\omega}$ for some $x$, whose length and density we can infer from $w$. The next theorem gives an upper bound on the waiting time for $x$, both in terms of the length of the non-periodic prefix $u$, and in the number of times a factor can occur before we can be sure that we have essentially found the periodic factor $x$ (up to rotation).

\begin{theorem}\label{thm:upperbounds}
Let $w\in {\cal L}$ and $v= \flipext^{\omega}(w)$. Let us write $v=ux^{\omega}$, with $|x|=\iota(w)$ and $x$ not a suffix of $u$. Let $\iota = \iota(w)$, $\kappa=\kappa(w),$ and $m = \left\lceil\frac{|w|}{\iota} \right\rceil$. Then

\begin{enumerate}
\item $|u| \leq ({\iota \choose \kappa} - 1)m\iota$, and
\item if for some $y\in \{0,1\}^{\iota}$, it holds that $y^{m+1}$ occurs with starting position $j>|w|$, then $y$ is a rotation of $x$.
\end{enumerate}
\end{theorem}

\begin{proof}
{\em 1.:} Assuming {\em 2.}, then every $\iota$-length factor $y$ which is not the final period can occur at most $m$ times consecutively. By Cor.~\ref{coro:lexicographic-order-of-the-factors}, consecutive non-equal factors in the iota-factorization of $v$ are lexicographically decreasing, so no factor $y$ can reoccur again once it has been replaced by another factor. By Theorem~\ref{thm:ult-periodic}, the density of each factor is $\kappa$. There are at most ${\iota \choose \kappa}$ such $y$ which are lexicographically smaller than $\pref_{\iota}(w)$, and each of these has length $\iota$.

{\em 2.:} By Lemma~\ref{lemma:lengthw}, in order to produce the next character of $v$, the operation \flipext\ needs to access only the last $|w|$ many characters of the current word. After $m+1$ repetitions of $u$, it holds that the $|w|$-length factor ending at position $i$ is equal to the $|w|$-length factor at position $i-\iota$, which proves the claim. 
\end{proof}

The following lemma motivates our interest in infinite words of the form $\flipext^{\omega}(w).$
 It says that $\flipext^{\omega}(w)$ is the prefix normal word
with the maximum number of $1$'s in each prefix among all prefix normal words having $w$ as prefix.

\begin{lemma}\label{lemma:densest}
Let $w\in {\cal L}$, $v= \flipext^{\omega}(w)$, and let $z \in {\cal L}$ such that $pref_{|w|}(z) = w$. Then for every $i = 1, 2, \dots$, we have $P_v(i) \geq P_z(i).$
\end{lemma}

\begin{proof} By contradiction, let $i > |w|$ be the smallest integer such that $P_v(i) < P_z(i)$. Then, by the minimality of $i$, we have $P_v(i-1) \geq P_z(i-1)$, hence $v_i = 0$ and $z_i = 1.$ The definition of the operation $\flipext$ together with $v_i = 0$ implies the existence of some $j > 0$ such that $P_v(j+1) = |v_{i-j} \cdots v_{i-1}|_1$ by Fact~\ref{basicfacts} {\em (iv)}, for otherwise we would have $v_i = 1.$ By the minimality of $i$ it must also hold that $P_z(j+1) \leq P_v(j+1)$. Let us write $v' = v_{i-j} \cdots v_{i-1}$ and $z' = z_{i-j} \cdots z_{i-1}$. Now assume that $|z'|_1 \geq |v'|_1$. Since $|v'|_1 = P_v(j+1) \geq P_z(j+1) \geq |z'|_1$, this implies $P_v(j+1) = P_z(j+1)$. But then $P_z(j+1) = P_v(j+1) = |v'|_1 < |z'|_1 + 1 = |z'z_i|_1$, in contradiction to $z$ being prefix normal. So we have $|z'|_1 < |v'|_1$. Once more by the minimality of $i$, it also holds that $P_v(i-j-1) \geq P_z(i-j-1)$, leading to
\[ P_v(i-1) = P_v(i-j-1) + |v'|_1 > P_z(i-j-1) + |z'|_1 = P_z(i-1),\]

which implies $P_v(i) \geq P_z(i)$, contradicting the initial assumption, and completing the proof. 
\end{proof}

\section{Conclusion}\label{sec:conclusion}

We presented a new recursive generation algorithm for prefix normal words of fixed length. The algorithm can also be used to generate all prefix normal words sharing the same critical prefix, thus serving as an aid for counting these words. The algorithm can generate the words either in lexicographic, or in a (combinatorial) Gray-code order.

We introduced infinite prefix normal words, and gave some results on the infinite extension of finite prefix normal words generated by a modified version of our algorithm. We found that the minimum prefix density, as well as its length, are important parameters of infinite prefix normal words. This fact allows us to make predictions about the structure of this infinite word, based on the starting prefix. A general investigation of infinite prefix normal words will be the subject of future research.

\section*{Acknowledgements}
We wish to thank two anonymous referees who read our paper very carefully and  contributed to improving  its presentation.

\bibliographystyle{abbrv}

\end{document}